%% file: main.tex
\documentclass[runningheads]{llncs}
\usepackage{graphicx}
\input{packages}
\input{macros}

\bibliography{ref}

\begin{document}
\title{DeepAbstract: Neural Network Abstraction for Accelerating Verification}
\author{Pranav Ashok\inst{1} \and
Vahid Hashemi\inst{2} \and
Jan K\v{r}et\'{i}nsk\'{y}\inst{1} \and
Stefanie Mohr\inst{1}
}
\authorrunning{P. Ashok et al.}
\institute{Technical University of Munich, Germany \and
Audi AG, Germany
}
\maketitle              %
\begin{abstract}
While abstraction is a classic tool of verification to scale it up, it is not used very often for verifying neural networks. However, it can help with the still open task of scaling existing algorithms to state-of-the-art network architectures. 
We introduce an abstraction framework applicable to fully-connected feed-forward neural networks based on clustering of neurons that behave similarly on \emph{some} inputs. 
For the particular case of ReLU, we additionally provide error bounds incurred by
the abstraction.
We show how the abstraction reduces the size of the network, while
preserving its accuracy, and how verification results on the abstract
network can be transferred back to the original network.

\end{abstract}
\input{1_introduction}
\input{2_preliminaries}
\input{3_abstraction}

\input{4_verification}

\input{5_experiments}

\section{Conclusion}

We have presented an abstraction framework for feed-forward neural networks using ReLU activation units.
Rather than just syntactic information, it reflects the semantics of the neurons, via our concept of I/O-similarity on experimental values.
In contrast to compression-based frameworks, the abstraction mapping between the original neurons and the abstract neurons allows for transferring verification proofs (transferring counterexamples is trivial), allowing for abstraction-based verification of neural networks. 

While we have demonstrated the potential of the new abstraction approach by a proof-of-concept implementation, its practical applicability relies on several next steps.
Firstly, I/O-similarity with the Euclidean distance ignores even any linear dependencies of the I/O-vectors; I/O-similarity with e.g.\ principal component analysis thus might yield orders of magnitude smaller abstractions, scaling to more realistic networks.
Secondly, due to the correspondence between the proofs, CEGAR could be employed: one can refine those neurons where the transferred constraints in the proof become too loose.
Besides, it is also desirable to extend the framework to other architectures, such as convolutional neural networks.

\printbibliography

\input{appendix}

\end{document}

%% file: packages.tex
\usepackage[utf8]{inputenc}
\usepackage[style=alphabetic,backend=bibtex]{biblatex}
\usepackage{lmodern}
\usepackage[final]{microtype}
\usepackage{ellipsis, mparhack, ragged2e} %
\usepackage[l2tabu, orthodox]{nag} %
\usepackage{amsmath,amsfonts,amssymb}
\usepackage{xcolor}

\usepackage{hyperref}

\usepackage{caption}
\usepackage{subcaption}

\usepackage{algorithm}
\usepackage[noend]{algpseudocode}

\usepackage{tikz}
\usetikzlibrary{positioning}

\usepackage{pgfplots}
\usepackage{booktabs}

\usepackage{cancel}

%% file: macros.tex
\newcommand{\todop}[1]{\textcolor{red}{{\bf P:} #1}}
\newcommand{\todos}[1]{\textcolor{green}{{\bf S:} #1}}
\newcommand{\todoj}[1]{\textcolor{blue}{\bf J: #1}}
\newcommand{\todo}[1]{\textcolor{purple}{#1}}
\newcommand{\improve}[1]{\textcolor{purple}{\textbf{Improve: } #1}}

\renewcommand{\todop}[1]{}
\renewcommand{\todos}[1]{}
\renewcommand{\todoj}[1]{}
\renewcommand{\todo}[1]{}
\renewcommand{\improve}[1]{}

\newcommand{\R}{\mathbb{R}}

\newcommand{\B}{\mathcal{B}}

\newcommand{\inputspace}{\mathcal{X}}
\newcommand{\inputset}{X}

\newcommand{\accumerror}{\vec{err}}
\newcommand{\perturb}{\vec\delta}

\newcommand{\NN}{D}
\newcommand{\matr}[2]{#1^{(#2)}}
\newcommand{\weight}[2]{w^{(#1)}_{#2}}
\newcommand{\vecz}{\mathbf{z}}
\newcommand{\z}{z}
\newcommand{\vech}{\mathbf{h}}
\newcommand{\h}{h}
\newcommand{\vecb}{\mathbf{b}}
\renewcommand{\b}{b}
\newcommand{\veca}{\mathbf{a}}
\renewcommand{\a}{a}

\newcommand{\relu}[1]{ReLU(#1)}

\newcommand{\numclusters}{K_L}

\newcommand{\lorig}{l}
\newcommand{\uorig}{u}
\newcommand{\lover}{\hat\lorig}
\newcommand{\uover}{\hat\uorig}
\newcommand{\lpre}{\tilde{l}}
\newcommand{\upre}{\tilde{u}}

\newcommand{\xpre}{\tilde{z}}
\newcommand{\zpost}{\tilde{\tilde{\vecz}}}
\newcommand{\zpre}{\tilde{\vecz}}

\newcommand{\uaccdelta}{\delta^{u}_{acc}}
\newcommand{\laccdelta}{\delta^{l}_{acc}}

%% file: 1_introduction.tex
\section{Introduction}\label{sec:intro}

\paragraph{Neural networks (NN)} are successfully used to solve many hard problems reasonably well in practice.
However, there is an increasing desire to use them also in safety-critical settings, such as perception in autonomous cars \cite{chen2017multi}, where reliability has to be on a very high level and that level has to be guaranteed, preferably by a rigorous proof.
This is a great challenge, in particular, since NN are naturally very susceptible to adversarial attacks, as many works have demonstrated in the recent years \cite{papernot2016limitations,akhtar2018threat,dong2018boosting,DBLP:journals/tec/SuVS19}.%
Consequently, various verification techniques for NN are being developed these days. 
Most verification techniques focus on proving robustness of the neural networks \cite{DBLP:conf/atva/ChengNR17,DBLP:conf/atva/Ehlers17,DBLP:conf/cav/HuangKWW17,DBLP:conf/cav/KatzBDJK17,gehr2018ai2,DBLP:conf/iclr/SinghGPV19}, i.e. for a classification task, when the input is perturbed by a small $\varepsilon$, the resulting output should be labeled the same as the output of the original input. 
Reliable analysis of robustness is computationally extremely expensive and verification tools struggle to scale when faced with real-world neural networks \cite{dvijotham2018dual}. 

\paragraph{Abstraction} \cite{DBLP:journals/toplas/ClarkeGL94, DBLP:conf/cav/ClarkeGJLV00} is one of the very classic techniques used in formal methods to obtain more understanding of a system as well as more efficient analysis.
Disregarding details irrelevant to the checked property allows for constructing a smaller system with a similar behaviour.
Although abstraction-based techniques are ubiquitous in verification, improving its scalability, such ideas have not been really applied to the verification of NN, except for a handful of works discussed later.

In this paper, we introduce an abstraction framework for NN.
In contrast to syntactic similarities, such as having similar weights on the edges from the previous layer \cite{DBLP:conf/icpr/ZhongYZ18}, our aim is to provide a behavioural, semantic notion of similarity, such as those delivered by predicate abstraction, since such notions are more general and thus more powerful.
Surprisingly, this direction has not been explored for NN.
One of the reasons is that the neurons do not have an explicit structure like states of a program that are determined by valuations of given variables. 
What are actually the values determining neurons in the network?

Note that in many cases, such as recognition of traffic signs or numbers, there are finitely many (say $k$) interesting data points on which and on whose neighbourhood the network should work well.
Intuitively, these are the key points that determine our focus, our scope of interest.
Consequently, we propose the following equivalence on neurons.
We evaluate the $k$ inputs, yielding for each neuron a $k$-tuple of its activation values.
This can be seen as a vector in $\mathbb R^k$. 
We stipulate that two neurons are similar if they have similar vectors, i.e, very close to each other.
To determine reasonable equivalence classes over the vectors, we use the machine-learning technique of k-means clustering \cite{hastie2009elements}.
While other techniques, e.g. principal component analysis \cite{DBLP:books/lib/Bishop07}, might also be useful, simple clustering is computationally cheap and returns reasonable results.
To summarize in other words, in the lack of structural information about the neurons, we use empirical behavioural information instead.

\paragraph{Applications}
Once we have a way of determining similar neurons, we can merge each equivalence class into a single neuron and obtain a smaller, abstracted NN.
There are several uses of such an NN.
Firstly, since it is a smaller one, it may be preferred in practice since, generally, smaller networks are often more robust, smoother, and obviously less resource-demanding to run \cite{compression-survey}.
Note that there is a large body of work on obtaining smaller NN from larger ones, e.g. see \cite{compression-survey, DBLP:journals/pieee/DengLHSX20}.
Secondly, and more interestingly in the safety-critical context, we can use the smaller abstract NN to obtain a guaranteed solution to the original problem (verifying robustness or even other properties) in two distinct ways:

\begin{enumerate}
	\item The smaller NN could replace the original one and could be easier to verify, while doing the same job (more precisely, the results can be $\varepsilon$-different where we can compute an upper bound on $\varepsilon$ from the abstraction).
	\item We can analyze the abstract NN more easily as it is smaller and then transfer the results (proof of correctness or a counterexample) to the original one, provided the difference $\varepsilon$ is small enough.
\end{enumerate}
The latter corresponds to the classic abstraction-based verification scenario.
For each of these points, we provide proof-of-concept experimental evidence of the method's potential.
\pagebreak
\paragraph{Our contribution} is thus the following:
\begin{itemize}

	\item We propose to explore the framework of abstraction by clustering based on experimental data. For feed-forward NN with ReLU, we provide error bounds.

	\item We show that the abstraction is also usable for compression.
	The reduction rate grows with the size of the original network, while the abstracted NN is able to achieve almost the same accuracy as the original network.

	\item We demonstrate the verification potential of the approach: (i) In some cases where the large NN was not analyzable (within time-out),
	we verified the abstraction using existing tools; for other NN, we could reduce verification times from thousands to hundreds of seconds.
	(ii) We show how to transfer a proof of robustness by a verification tool DeepPoly \cite{deeppoly} on the abstract NN to a proof on the original network, whenever the clusters do not have too large radii.	
\end{itemize}

\paragraph{Related work}

In contrast to compression techniques, our abstraction provides a mapping between original neurons and abstract neurons, which allows for transferring the claims of the abstract NN to the original one, and thus its verification.

The very recent work \cite{katzlatest} suggests an abstraction, which is based solely on the sign of the effect of increasing a value in a neuron. 
While we can demonstrate our technique on e.g. 784 dimension input (MNIST) and  work with general networks, \cite{katzlatest}  is demonstrated only on the Acas Xu \cite{acasxu} networks which have 5 dimensional input; 
our approach handles thousands of nodes while the benchmark used in \cite{katzlatest} is of size 300. 
Besides, we support both classification and regression networks.
Finally, our approach is not affected by the number of outputs, whereas the \cite{katzlatest} grows exponentially with respect to number of outputs.

\cite{DBLP:conf/nips/PrabhakarA19} produces so called Interval Neural Networks containing intervals instead of single weights and performs abstraction by merging these intervals. 
However, they do not provide a heuristic for picking the intervals to merge, but pick randomly. 
Further, the results are demonstrated only on the low-dimensional Acas Xu networks.

Further, \cite{DBLP:conf/bmvc/SrinivasB15} computes a similarity measure between incoming weights and then starts merging the most similar ones. 
It also features an analysis of how many neurons to remove in order to not lose too much accuracy. 
However, it does not use clustering on the semantic values of the activations, but only on the syntactic values of the incoming weights, which is a very local and thus less powerful criterion. 
Similarly, \cite{DBLP:conf/icpr/ZhongYZ18} clusters based on the incoming weights only and does not bound the error.
\cite{DBLP:journals/corr/HanMD15} clusters weights in contrast to our activation values) using the k-means clustering algorithm. However, the focus is on weight-sharing and reducing memory consumption, treating neither the abstraction mapping nor verification. 

Finally, abstracting neural networks for verification purposes  was first proposed by \cite{cav10abstractionrefinement}, transforming the networks into Boolean constraints.

%% file: 2_preliminaries.tex
\section{Preliminaries}

We consider simple feedforward neural networks, denoted by $\NN$, consisting of one input layer, one output layer and one or more hidden layers. 
The layers are numbered $1, 2, \dots, L$ with $1$ being the \emph{input layer}, $L$ being the \emph{output layer} and $2, \dots, L - 1$ being the \emph{hidden layers}.
Layer $\ell$ contains $n_\ell$ \emph{neurons}.
A neuron is a computation unit which takes an input $\h \in \R$, applies an \emph{activation function} $\phi: \R \to \R$ on it and gives as output $\z = \phi(\h)$.
Common activation functions include tanh, sigmoid or ReLU \cite{maas2013rectifier}, however we choose to focus on ReLU for the sake of simplicity, where ReLU($x$) is defined as $\max(0, x)$.
Neurons of one layer are connected to neurons of the previous and/or next layers by means of weighted connections.
Associated with every layer $\ell$ that is not an output layer is a \emph{weight matrix} $\matr{W}{\ell} = (\weight{\ell}{i,j}) \in \R^{n_{\ell+1} \times n_\ell}$
where $\weight{\ell}{i,j}$ gives the weights of the connections to the $i^{th}$ neuron in layer $\ell+1$ from the $j^{th}$ neuron in layer $\ell$. 
We use the notation $\matr{W}{\ell}_{i,*} = [\weight{\ell}{i,1}, \dots, \weight{\ell}{i,n_\ell}]$ to denote the incoming weights of neuron $i$ in layer $\ell+1$ and 
$\matr{W}{\ell}_{*,j} = [\weight{\ell}{1,j}, \dots, \weight{\ell}{n_{\ell+1},j}]^\intercal$ to denote the outgoing weights of neuron $j$ in layer $\ell$. 
Note that $\matr{W}{\ell}_{i,*}$ and $\matr{W}{\ell}_{*,j}$ correspond to the $i^{th}$ row and $j^{th}$ column of $\matr{W}{\ell}$ respectively.
The input and output of a neuron $i$ in layer $\ell$ is denoted by $\h^{(\ell)}_i$ and $\z^{(\ell)}_i$ respectively. 
We call $\vech^{\ell} = [\h^{(\ell)}_1, \dots, \h^{(\ell)}_{n_\ell}]^\intercal$ the vector of \emph{pre-activations} of layer $\ell$ 
and $\vecz^{\ell} = [\z^{(\ell)}_1, \dots, \z^{(\ell)}_{n_\ell}]^\intercal$ the vector of \emph{activations} of layer $\ell$, where $\z^{(\ell)}_i = \phi^{(\ell)}(\h^{(\ell)}_i)$. 
A vector $\vecb^{(\ell)} \in \R^{n_\ell}$ called \emph{bias} is also associated with all hidden layers $\ell$.

In a feedforward neural network, information flows strictly in one direction: from layer $\ell_m$ to layer $\ell_n$ where $\ell_m < \ell_n$. 
For an $n_1$-dimensional input $\vec{x} \in \inputspace$ from some input space $\inputspace \subseteq \R^{n_1}$, the output $\vec{y} \in \R^{n_L}$ of the neural network $\NN$, also written as $\vec{y} = \NN(\vec{x})$ is iteratively computed as follows:
\begin{align}
\vech^{(0)} &= \vec{x}\notag\\
\vech^{(\ell+1)} &= \matr{W}{\ell} \vecz^{(\ell)} + \vecb^{(\ell+1)} \label{eq:forward-h}\\
\vecz^{(\ell+1)} &= \phi (\vech^{(\ell+1)}) \label{eq:forward-z}\\
\vec{y} &= \vecz^{(L)}\notag
\end{align}
where $\phi(x)$ is the column vector obtained on applying $\phi$ component-wise to $\vec{x}$.
We sometimes write $\vecz^{(\ell)}(\vec{x})$ to denote the output of layer $\ell$ when $\vec{x}$ is given as input to the network.

We define a \emph{local robustness} query to be a tuple $Q = (\NN, \vec{x}, \perturb)$ for some network $\NN$, input $\vec{x}$ and perturbation $\perturb \in \R^{\lvert \vec{x} \rvert}$ and call $\NN$ to be robust with respect to $Q$ if $\forall \vec{x'} \in [\vec{x}-\perturb, \vec{x}+\perturb]\,:\, \NN(\vec{x'}) = \NN(\vec{x})$. In this paper, we only deal with local robustness.

%% file: 3_abstraction.tex
\section{Abstraction} \label{sec:abstraction}

In classic abstraction, states that are similar with respect to a property of interest are merged for analysis.
In contrast, for NN,
it is not immediately clear %
which neurons to merge and what similarity means.
Indeed, neurons are not actually states/configurations of the system; as such, neurons, as opposed to states with values of variables, do not have inner structure.
Consequently, identifying and dropping irrelevant information (part of the structure) becomes more challenging. 
We propose to merge neurons which compute a similar function \emph{on some set $\inputset$ of inputs}, i.e., for each input $x\in\inputset$ to the network, they compute $\varepsilon$-close values. 
We refer to this as I/O-similarity. Further, we choose to merge neurons only within the same layer to keep the analysis and implementation straightforward.

In Section \ref{sec:merge}, we show a straightforward  way to merge neurons in a way that is sensible if they are I/O-similar.
In Section \ref{sec:clustering-abstraction}, we give a heuristic for partitioning neurons into classes according to their I/O-similarity.
While this abstraction idea is not limited to verification of robustness, it preserves the robustness of the original network particularly well, as seen in the experiments in Section~\ref{sec:experiments}.

\subsection{Merging I/O-similar neurons}\label{sec:merge}

I/O-similar neurons can be merged easily without changing the behaviour of the NN too much.
First, we explain the procedure on an example.

\begin{figure}
    \centering
    \begin{subfigure}{.5\textwidth}
        \centering
        \begin{tikzpicture}[every node/.style={circle}, font=\scriptsize]
        \node[draw] (1) at (0,0.8) {1};
        \node[draw] (2) at (0,-0.8) {2};
        \node[draw] (3) at (1.6,1.3) {3};
        \node[draw] (4) at (1.6,0) {4};
        \node[draw] (5) at (1.6,-1.3) {5};
        \node[draw] (6) at (3.2, 0) {6};
        
        \draw[->] (1) -- (3) node[midway,above] {$w_1$};
        \draw[->] (1) -- (4) node[pos=0.35,above,yshift=-3pt] {$w_2$};
        \draw[->] (1) -- (5) node[pos=0.2,left] {$w_3$};
        \draw[->] (2) -- (3) node[pos=0.75,right] {$w_4$};
        \draw[->] (2) -- (4) node[pos=0.9,below,yshift=2pt] {$w_5$};
        \draw[->] (2) -- (5) node[midway,below,yshift=2pt] {$w_6$};
        \draw[->] (3) -- (6) node[midway,above] {$w_7$};
        \draw[->] (4) -- (6) node[near start,above,yshift=-3pt] {$w_8$};
        \draw[->] (5) -- (6) node[midway,below] {$w_9$};
        \end{tikzpicture}
        \caption{Original network}
        \label{fig:merge-before}
    \end{subfigure}%
    \begin{subfigure}{.5\textwidth}
        \centering
        \begin{tikzpicture}[every node/.style={circle}, font=\scriptsize]
        \node[draw] (1) at (0,0.8) {1};
        \node[draw] (2) at (0,-0.8) {2};
        \node[draw] (3) at (1.6,1.3) {3};
        \node[draw] (4) at (1.6,0) {4};
        \node[draw,dotted] (5) at (1.6,-1.3) {\textcolor{gray}{5}};
        \node[draw] (6) at (3.2, 0) {6};
        
        \draw[->] (1) -- (3) node[midway,above] {$w_1$};
        \draw[->] (1) -- (4) node[pos=0.35,above,yshift=-3pt] {$w_2$};
        \draw[->] (2) -- (3) node[pos=0.75,right] {$w_4$};
        \draw[->] (2) -- (4) node[pos=0.8,below,yshift=2pt] {$w_5$};
        \draw[->] (3) -- (6) node[midway,above] {$w_7$};
        \draw[->] (4) -- (6) node[midway,below,yshift=8pt] {$w_8+w_9$};
        \end{tikzpicture}
        \caption{Network after merging neurons 4 and 5}
        \label{fig:merge-after}
    \end{subfigure}
    \caption{Before and after merge: neuron 4 is chosen as a representative of both 4 and 5. On merging, the incoming weights of neuron 5 are deleted and its outgoing weight is added to the outgoing weight of neuron 4.}
    \label{fig:merge}
\end{figure}
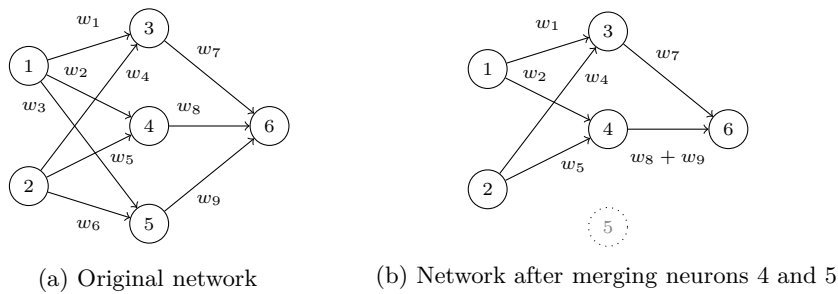

\begin{example}
Consider the network shown in Figure \ref{fig:merge-before}. 
The network contains 2 input neurons and 4 ReLU neurons. For simplicity, we skip the bias term in this example network.
Hence, the activations of the neurons in the middle layer are given as follows: 
$\z_3 = \relu{w_1 \z_1 + w_4 \z_2}$, $\z_4 = \relu{w_2 \z_1 + w_5 \z_2}$, $\z_5 = \relu{w_3 \z_1 + w_6 \z_2}$;
and the output of neuron 6 is $\z_6 = \relu{w_7 \z_3 + w_8 \z_4 + w_9 \z_5}$. 
Suppose that for all inputs in the dataset, the activations of neurons 4 and 5 are `very' close, denoted by $\z_4 \approx \z_5$. Then, $\z_6 = \relu{w_7 \z_3 + w_8 \z_4 + w_9 \z_5} $.

Since neurons 4 and 5 behave similarly, we abstract the network by merging the two neurons as shown in Figure \ref{fig:merge-after}. 
Here, neuron 4 is chosen as a representative of the ``cluster'' containing neurons 4 and 5, and the outgoing weight of the representative is set to the sum of outgoing weights of all the neurons in the cluster. 
Note that the incoming weights of the representative do not change. 
In the abstracted network, the activations of the neurons in the middle layer are now given by $\tilde\z_3 = \relu{w_1 \tilde\z_1 + w_4 \tilde\z_2} = \z_3$ and $\tilde\z_4 = \relu{w_2 \tilde\z_1 + w_5 \tilde\z_2} = \z_4$ with neuron 5 being removed. The output of neuron 6 is therefore $\tilde\z_6 = \relu{w_7 \tilde\z_3 + (w_8 + w_9) \tilde\z_4} = \relu{w_7 \z_3 + (w_8 + w_9) \z_4} = \relu{w_7 \z_3 + w_8\z_4 + w_9\z_4} \approx \z_6$, which illustrates that merging preserves the behaviour of the network.
\end{example}

Formally, the process of merging two neurons $p$ and $q$ belonging to the same layer $\ell$ works as follows. We assume, without loss of generality, that $p$ is retained as the representative.
First, the abstract network $\tilde\NN$ is set to the original network $\NN$. 
Next, $\matr{\tilde W}{\ell-1}$ is set to $\matr{W}{\ell-1}$ with the $q^{th}$ row deleted. 
Further, we set the outgoing weights of the representative $p$ to the sum of outgoing weights of $p$ and $q$, $\matr{\tilde W}{\ell}_{*,p} = \matr{W}{\ell}_{*,p} + \matr{W}{\ell}_{*,q}$.
This procedure is naturally extendable to merging multiple I/O-similar neurons.
It can be applied repeatedly until all desired neurons are merged.
For the interested reader, the correctness proof and further technical details are made available in Appendix \ref{app:hard-merge-correctness}.

\begin{proposition}[Sanity Check]
	If for neurons $p$ and $q$, for all considered inputs $x \in \inputset$ to the network $\NN$, $\z_p = \z_q$, then the network $\tilde\NN$ produced as described above, in which $p$ and $q$ are merged by removing $q$ and letting $p$ serve as their representative, and by setting $\matr{\tilde W}{\ell}_{*,p} = \matr{W}{\ell}_{*,p} + \matr{W}{\ell}_{*,q}$, will have the same output as $\NN$ on all inputs $x \in \inputset$. In other words, $\forall x \in \inputset~\NN(x) = \tilde\NN(x)$.
\end{proposition}

\subsection{Clustering-based Abstraction} \label{sec:clustering-abstraction}

\begin{algorithm}[t]
    \caption{Abstract network \NN\ with given clustering $\numclusters$} \label{alg:run-abstraction}
    \begin{algorithmic}[1]
        \Procedure{Abstract}{$\NN, \inputset, \numclusters$}
        \State $\tilde\NN \gets \NN$
        \For{$\ell \gets 2, \dots, L-1$}
        \State $A \gets \{ \veca_i^{(\ell)} \mid \veca^{(\ell)}_i = [\tilde\z^{(\ell)}_i(x_1), \dots, \tilde\z^{(\ell)}_i(x_N)] \text{ where } x_i \in \inputset \}$
        \State $\mathcal{C} \gets \textsc{kmeans}(A, \numclusters(\ell))$
        \For{$C \in \mathcal{C}$}
        \State $\matr{\tilde{W}}{\ell}_{*, rep(C)} \gets \sum_{i \in C} \matr{W}{\ell}_{*, i}$
        \EndFor
        \State \textbf{delete} $C\setminus\{rep(C)\}$ from $\tilde\NN$
        \EndFor
        \Return $\tilde\NN$
        \EndProcedure
    \end{algorithmic}
\end{algorithm}

In the previous section, we saw that multiple I/O-similar neurons can be merged to obtain an abstract network behaving similar to the original network. 
However, the quality of the abstraction depends on the choice of neurons used in the merging.
Moreover, it might be beneficial to have multiple groups of neurons that are merged separately.
While multiple strategies can be used to identify such groups, in this section, we illustrate this on one of them --- the unsupervised learning approach of \textit{$k$-means clustering} \cite{DBLP:books/lib/Bishop07}, as a proof-of-concept.

Algorithm \ref{alg:run-abstraction} describes how the approach works in general.
It takes as input the original (trained) network $\NN$, an input set $\inputset$ and a function $\numclusters$, which for each layer gives the number of clusters to be identified in that layer.
Each $x \in \inputset$ is input into $\tilde\NN$ and for each neuron $i$ in layer $\ell$, an $\lvert\inputset\rvert$-dimensional vector of observed activations $\veca^{(\ell)}_i = [\z^{(\ell)}_i(x_1), \dots, \z^{(\ell)}_i(x_{\lvert\inputset\rvert})]$ is constructed.
These vectors of activations, one for each neuron, are collected in the set $A$.
We can now use the $k$-means algorithm on the set $A$ to identify $\numclusters(\ell)$ clusters.
Intuitively, $k$-means aims to split the set $A$ into $\numclusters(\ell)$ clusters such that the pairwise squared deviations of points in the same cluster is minimized.
Once a layer is clustered, the neurons of each cluster are merged and the neuron closest to the centroid of the respective cluster, denoted by $rep(C)$ in the pseudocode, is picked as the cluster representative.
As described in Section \ref{sec:merge}, the outgoing connections of all the neurons in a cluster are added to the representative neuron of the cluster and all neurons except the representative are deleted.

While Algorithm \ref{alg:run-abstraction} describes the clustering procedure, it is still a challenge to find the right $\numclusters$.
In Algorithm \ref{alg:identify-clusters}, we present one heuristic to identify a good set of parameters for the clustering.
It is based on the intuition that merging neurons closer to the output layer impacts the network accuracy the least, as the error due to merging is not multiplied and propagated through multiple layers.
The overarching idea is to search for the best $k$-means parameter, $\numclusters(\ell)$, for each layer $\ell$ starting from the first hidden layer to the last hidden layer, while making sure that the merging with the said parameter ($\numclusters$) does not drop the accuracy of the network beyond a threshold $\alpha$.

\begin{algorithm}[t]
    \caption{Algorithm to identify the clusters} \label{alg:identify-clusters}
    \begin{algorithmic}[1]
        \Procedure{Identify-clusters}{$\NN, \inputset, \alpha$}
        \State $\tilde\NN \gets \NN$
        \For{$\ell \gets 2,...,L-1$}  \Comment{Loops through the layers}
        \If{$accuracy(\tilde\NN) > \alpha$}
        \State $\numclusters(\ell) \gets \textsc{BinarySearch}(\tilde\NN, \alpha, \ell)$ \Comment{Finds optimal number of clusters}
        \State $\tilde\NN \gets \textsc{Abstract}(\tilde\NN, \inputset, \numclusters)$
        \EndIf
        \EndFor
        \State \textbf{return} $\numclusters$
        \EndProcedure
    \end{algorithmic}
\end{algorithm}

The algorithm takes a trained network $\NN$ as input along with an input set $\inputset$ and a parameter $\alpha$, the lower bound on the accuracy of the abstract network.
The first hidden layer ($\ell = 2$) is picked first and $k$-means clustering is attempted on it.
The parameter $\numclusters(\ell)$ is discovered using the \textsc{BinarySearch} procedure which searches for the lowest $k$ such that the accuracy of the network abstracted with this parameter is the highest. 
We make a reasonable assumption here that a higher degree of clustering (i.e. a small $k$) leads to a higher drop in accuracy.
Note that this might cause the \textsc{BinarySearch} procedure to work on a monotone space and we might not exactly get the optimal. However, in our experiments, the binary search turned out to be a sufficiently good alternative to brute-force search.
The algorithm ensures that merging the clusters as prescribed by $\numclusters$ does not drop the accuracy of the abstracted network below $\alpha$.\footnote{Naturally, the parameter $\alpha$ has to be less than or equal to the accuracy of $\NN$}
This process is now repeated on $\tilde\NN$ starting with the next hidden layer.
Finally, $\numclusters$ is returned, ready to be used with Algorithm \ref{alg:run-abstraction}.

Now we present two results which bound the error induced in the network due to abstraction. 
The first theorem applies to the case where we have clustered groups of I/O-similar neurons in each layer for the set $\inputset$ of network inputs.

Let for each neuron $i$, $\veca_i = [\z_i(x_1), \dots, \z_i(x_N)]$ where $x_j \in \inputset$, and let $\tilde\NN$ = \textsc{Abstract($\NN, \inputset, \numclusters$)} for some given $\numclusters$. Define $\vec\epsilon^{(\ell)}$, the maximal distance of a neuron from the respective cluster representative, as
\begin{align}
\vec\epsilon^{(\ell)} = [\epsilon^{(\ell)}_1, \dots, \epsilon^{(\ell)}_{n_\ell}]^\intercal
\quad\quad\quad
\text{where}
\quad\quad\quad
\epsilon^{(\ell)}_i = \lVert\veca_i - \veca_{r_{C_i}}\rVert \label{eq:max-cluster-width}
\end{align}
where $\lVert \cdot \rVert$ denotes the Euclidean norm operator, $C_i$ denotes the cluster containing $i$ and $r_{C_i}$ denotes the representative of cluster $C_i$. Further, define the absolute error due to abstraction in layer $\ell$ as $\accumerror^{(\ell)} = \tilde\vecz^{(\ell)} - \vecz^{(\ell)}$.

\begin{theorem}[Clustering-induced error]\label{thm:accum-abs-error}
	If the accumulated absolute error in the activations of layer $\ell$ is given by $\accumerror^{(\ell)}$ and $\vec\epsilon^{(\ell+1)}$ denotes the the maximal distance of each neuron from their cluster representative (as defined in Eqn. \ref{eq:max-cluster-width}) of layer $\ell+1$, then the absolute error $\accumerror^{(\ell+1)}$ for all inputs $\vec{x} \in \inputset$ can be bounded by
	\[
	\lvert \accumerror^{(\ell+1)} \rvert \leq \lvert \matr{W}{\ell} \accumerror^{(\ell)} \rvert + \vec \epsilon^{(\ell+1)}
	\]
	and hence, the absolute error in the network output is given by $\accumerror^{(L)}$.
\end{theorem}

The second result considers the local robustness setting where we are interested in the output of the abstracted network when the input $\vec{x} \in \inputset$  is perturbed by $\delta \in \R^{\lvert x \rvert}$. 

\begin{theorem}\label{thm:robustness}
	If the inputs $\vec{x} \in \inputset$ to the abstract network $\tilde\NN$ are perturbed by $\perturb \in \R^{\lvert \vec{x} \rvert}$, then the absolute error in the network output due to both abstraction and perturbation denoted by $\accumerror_{total}$ is bounded for every $\vec{x} \in \inputset$ and is given by
	$$
	\lvert \accumerror_{total} \rvert \leq
	\lvert \matr{\tilde W}{L} \dots \matr{\tilde W}{1} \perturb \rvert + 
	\lvert \accumerror^{(L)} \rvert
	$$
	where $\matr{\tilde W}{\ell}$ is the matrix of weights from layer $\ell$ to $\ell+1$ in $\tilde\NN$, $L$ is the number of layers in $\tilde\NN$ and $\accumerror^{(L)}$ is the accumulated error due to abstraction as given by Theorem \ref{thm:accum-abs-error}.
\end{theorem}

In other words, these theorems allow us to compute the absolute error produced due to  the abstraction alone; or due to both (i) abstraction and (ii) perturbation of input. Theorem \ref{thm:robustness} gives us a direct (but na\"{i}ve) procedure to perform local robustness verification by checking if there exists an output neuron $i$ with a lower bound ($\tilde\NN_i(x) - (E_{total})_i$) greater than the upper bound ($\tilde\NN_j(x) + (E_{total})_j$) of all other output neurons $j$. The proofs of both theorems can be found in Appendix \ref{app:error-bounds}.

%% file: 4_verification.tex
\section{Lifting guarantees from abstract NN to original NN}\label{sec:verification}

In the previous section, we discussed how a large neural network could be abstracted and how the absolute error on the output could be calculated and even used for robustness verification.
However, the error bounds presented in Theorem \ref{thm:robustness} might be too coarse to give any meaningful guarantees.
In this section, we present a proof-of-concept approach for lifting verification results from the abstracted network to the original network.
While in general the lifting depends on the verification algorithm, as a demonstrative example, we show how to perform the lifting when using the verification algorithm DeepPoly \cite{deeppoly} and also how it can be used in conjunction with our abstraction technique to give robustness guarantees on the original network.

We now give a quick summary of DeepPoly. Assume that we need to verify that the network $\NN$ labels all inputs in the $\delta$-neighborhood of a given input $x \in \inputset$ to the same class; in other words, check if $\NN$ is locally robust for the robustness query $(\NN, \vec{x}, \delta)$.
DeepPoly functions by propagating the interval $[\vec{x}-\vec{\delta}, \vec{x}+\vec{\delta}]$ through the network with the help of abstract interpretation, producing over-approximations (a lower and an upper bound) of activations of each neuron.
The robustness query is then answered by checking if the lower bound of the neuron representing one of the labels is greater than the upper bounds of all other neurons. 
We refer the interested reader to \cite[Section 2]{deeppoly} for an overview of DeepPoly.
Note that the algorithm is sound but not complete.

If DeepPoly returns the bounds $\lpre$ and $\upre$ for the abstract network $\tilde\NN$, the following theorem allows us to compute $[\hat\lorig, \hat\uorig]$ such that $[\hat\lorig, \hat\uorig] \supseteq [\lorig, \uorig]$, where $[\lorig, \uorig]$ would have been the bounds returned by DeepPoly on the original network $\NN$.

\begin{theorem}[Lifting guarantees] \label{thm:monster-thm}
    Consider the abstraction $\tilde\NN$ obtained by applying Algorithm \ref{alg:run-abstraction} on a ReLU feedforward network $\NN$. Let $\lpre^{(\ell)}$ and $\upre^{(\ell)}$ denote the lower bound and upper bound vectors returned by DeepPoly for the layer $\ell$, and let $\matr{\tilde{W}}{\ell}_+ = \max(0, \matr{\tilde{W}}{\ell})$ and $\matr{\tilde{W}}{\ell}_- = \min(\matr{\tilde{W}}{\ell}, 0)$ denote the +ve and -ve entries respectively of its $\ell^{th}$ layer weight matrix. Let $\vec\epsilon^{(\ell)}$ denote the vector of maximal distances of neurons from their cluster representatives (as defined in Equation \ref{eq:max-cluster-width}), and let $\vec{x}$ be the input we are trying to verify for a perturbation $[-\vec\perturb, \vec\perturb]$. Then for all layers $\ell < L$, we can compute
    \begin{equation}
    	\openup\jot %
    	\begin{aligned}[t]
    		\uover^{(\ell)} &= \max\Bigg(0, \begin{aligned}
    			&\matr{\tilde{W}}{\ell-1}_+(\uover^{(\ell-1)}+\vec\epsilon^{(\ell-1)})\nonumber\\
    			&+\matr{\tilde{W}}{\ell-1}_-(\lover^{(\ell-1)}-\vec\epsilon^{(\ell-1)})\nonumber\\
    			&+\tilde{b}^{(\ell)}\nonumber\\
    		\end{aligned}\Bigg)\\
    	\end{aligned}
    \quad
    	\begin{aligned}[t]
    		\lover^{(\ell)} &= \max\Bigg(0, \begin{aligned}
    			&\matr{\tilde{W}}{\ell-1}_+(\lover^{(\ell-1)}-\vec\epsilon^{(\ell-1)})\nonumber\\
    			&+\matr{\tilde{W}}{\ell-1}_-(\uover^{(\ell-1)}+\vec\epsilon^{(\ell-1)})\nonumber\\
    			&+\tilde{b}^{(\ell)}\\
    		\end{aligned}\Bigg)\nonumber\\
    	\end{aligned}
    \end{equation}
    where $\uover^{(1)} = \upre^{(1)} = \uorig^{(1)} =  \vec{x} + \vec\perturb$ and
    $\lover^{(1)} = \lpre^{(1)} = \lorig^{(1)} = \vec{x} - \vec\perturb$ such that
    \[
    [\lover, \uover] \supseteq [\lorig, \uorig]
    \]
    where $[\lorig, \uorig]$ is the bound computed by DeepPoly on the original network.
    
    For output layer $\ell = L$, the application of the $\max(0, \cdot)$-function is omitted, the rest remains the same.

\end{theorem}

In other words, this theorem allows us to compute an over-approximation of the bounds computed by DeepPoly on the original network $\NN$ by using only the abstract network, thereby allowing a local robustness proof to be lifted from the abstraction to the original network. Note that while this procedure is sound, it is not complete since the bounds computed by Theorem \ref{thm:monster-thm} might still be too coarse. An empirical discussion is presented in Section \ref{sec:experiments}, an example of the proof lifting can be seen in Appendix \ref{app:example-lifting}, and the proof is given in Appendix \ref{app:lifting-proof}.

%% file: 5_experiments.tex
\section{Experiments}\label{sec:experiments}

We now analyze the potential of our abstraction. In particular, in Section \ref{sec:abstraction-results}, we look at how much we can abstract while still guaranteeing a high test accuracy for the abstracted network. Moreover, we present verification results of abstracted network, suggesting a use case where it replaces the original network. In Section \ref{sec:verification-results}, we additionally consider lifting of the verification proof from the abstracted network to the original network.

We ran experiments with multiple neural network architectures on the popular MNIST dataset \cite{lecun1998mnist}. %
We refer to our network architectures by the shorthand $L \times n$, for example ``$6\times100$'', to denote a network with L fully-connected feedforward hidden layers with $n$ neurons each, along with a separate input and output layers whose sizes depend on the dataset --- 784 neurons in the input layer and 10 in the output layer in the case of MNIST. Interested readers may find details about the implementation in Appendix \ref{app:implementation-details}.

\paragraph{Remark on Acas Xu}
We do not run experiments on the standard NN verification case study Acas Xu \cite{acasxu}.
The Acas Xu networks are very compact, containing only 6 layers with 50 neurons each.
The training/test data for these networks are not easily available, which makes it difficult to run our data-dependent abstraction algorithm.
Further, the network architecture cannot be scaled up to observe the benefits of abstraction, which, we conjecture, become evident only for large networks possibly containing redundancies.
Moreover, the specifications that are commonly verified on Acas Xu are not easily encodable in DeepPoly.

\subsection{Abstraction results} \label{sec:abstraction-results}
First, we generated various NN architectures by scaling up the number of neurons per layer as well as the number of layers themselves and trained them on MNIST. 
More information on the training process is available in Appendix \ref{app:training-details}. 
Then, we executed our clustering-based abstraction algorithm (Algorithm \ref{alg:run-abstraction}) on each trained network allowing for a drop in accuracy on a test dataset of at most 1\%.

\paragraph{Size of the abstraction}

Table \ref{tab:inc_number} gives some information about the quality of the abstraction - the extent to which we can abstract while sacrificing accuracy of at most 1\%.
We can see that increasing the width of a layer (number of neurons) while keeping the depth of the network fixed increases the number of neurons that can be merged, i.e. the reduction rate increases.
We conjecture that there is a minimum number of neurons per layer that are needed to simulate the behavior of the original network. 
On the other hand, interestingly, if the depth of the network is increased while keeping the width fixed, the reduction rate seems to hover around 15-20\%.

\begin{table}[t]
	\centering
	\hfill
	\setlength{\tabcolsep}{12pt}
	\caption{Reduction rate of abstracted neural networks with different architectures along with the drop in accuracy (measured on an independent test set). In the top half, the number of layers (depth) is varied and in the bottom half, the number of neurons per layer (width) is increased. This table shows that the clustering-based abstraction works better with wider networks.}
	\label{tab:inc_number}
	\vspace{1em}
	\begin{tabular}{lrr}
		\toprule
		\textbf{Network} & \textbf{Accuracy} & \textbf{Reduction} \\
		\textbf{Arch.}   &     \textbf{Drop (\%)} &      \textbf{Rate (\%)} \\ \midrule
		$3\times100$     &               0.40 &               15.5 \\
		$4\times100$     &              0.41 &               15.5 \\
		$5\times100$     &              0.21 &               21.2 \\
		$6\times100$     &               0.10 &               13.3 \\ \midrule
		$6\times50$      &               0.10 &                5.7 \\
		$6\times100$     &               0.10 &               13.3 \\
		$6\times200$     &               0.10 &               30.2 \\
		$6\times300$     &               0.20 &               39.9 \\
		$6\times1000$    &              0.01 &               61.7 \\ \bottomrule
	\end{tabular}
\end{table}

\begin{figure}
	\centering
	\begin{tikzpicture}
	\begin{axis}[
	width=0.95\textwidth,
	height=0.6\textwidth,
	symbolic x coords={Layer 1, Layer 2, Layer 3, Layer 4, Layer 5, Layer 6, DUMMY},
	xtick=data,
	ylabel={Neurons after abstraction},
	enlargelimits=0.05,
	legend style={at={(0.66,0.95)},
		anchor=north,legend columns=-1},
	ybar interval=0.6,
	]

	\addplot[color=blue,fill=blue!40!white]
	coordinates{(Layer 1,240)(Layer 2,123)(Layer 3,50)(Layer 4,33)(Layer 5,19)(Layer 6,29)(DUMMY,1)};
	\addplot [color=red,fill=red!40!white]
	coordinates {(Layer 1,136)(Layer 2,71)(Layer 3,32)(Layer 4,35)(Layer 5,31)(Layer 6,26)(DUMMY,1)};
	\addplot[color=green,fill=green!40!white]
	coordinates {(Layer 1,80)(Layer 2,77)(Layer 3,43)(Layer 4,24)(Layer 5,24)(Layer 6,19)(DUMMY,1)};
	\addplot[color=black,fill=black!40!white]
	coordinates{(Layer 1,48)(Layer 2,41)(Layer 3,45)(Layer 4,36)(Layer 5,25)(Layer 6,28)(DUMMY,1)};
	\legend{$6\times500$,$6\times200$,$6\times100$,$6\times50$}
	
	\coordinate (A) at (axis cs:Layer 1,100);
	\coordinate (B) at (axis cs:Layer 1,200);
	\coordinate (C) at (axis cs:Layer 1,50);
	\coordinate (D) at (axis cs:Layer 1,500);
	\coordinate (O1) at (rel axis cs:0,0);
	\coordinate (O2) at (rel axis cs:1,0);
	
	\draw [thick,green,sharp plot,dashed] (A -| O1) -- (A -| O2);
	\draw [thick,red,sharp plot,dashed] (B -| O1) -- (B -| O2);
	\draw [thick,black,sharp plot,dashed] (C -| O1) -- (C -| O2);

	\end{axis}
	\end{tikzpicture}
	\caption{Plot depicting the sizes of the abstract networks when initialized with 4 different architectures and after repetitively applying clustering-based abstraction on the layers until their accuracy on the test set is approximately 95\%. }\label{fig:repetitive-clustering}
\end{figure}
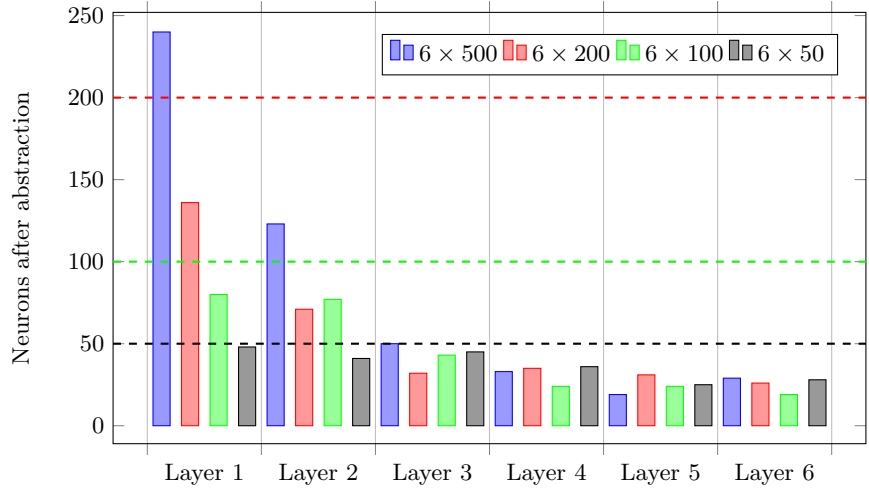

Figure \ref{fig:repetitive-clustering} demonstrates the potential of the clustering-based abstraction procedure in compressing the network.
Here, the abstraction is performed layer after layer from layer 1 to layer 6. We cluster as much as possible permitting the test accuracy of the network to drop by at most 1\%.
Unsurprisingly, we get more reduction in the later (closer to output) layers compared to the initial. 
We conjecture that this happens as the most necessary information is already processed and computed early on, and the later layers transmit low dimensional information.
Interestingly, one may observe that in layers 4, 5 and 6, all network architectures ranging from 50 to 500 neurons/layer can be compressed to an almost equal size around 30 nodes/layer.

\paragraph{Verifying the abstraction}

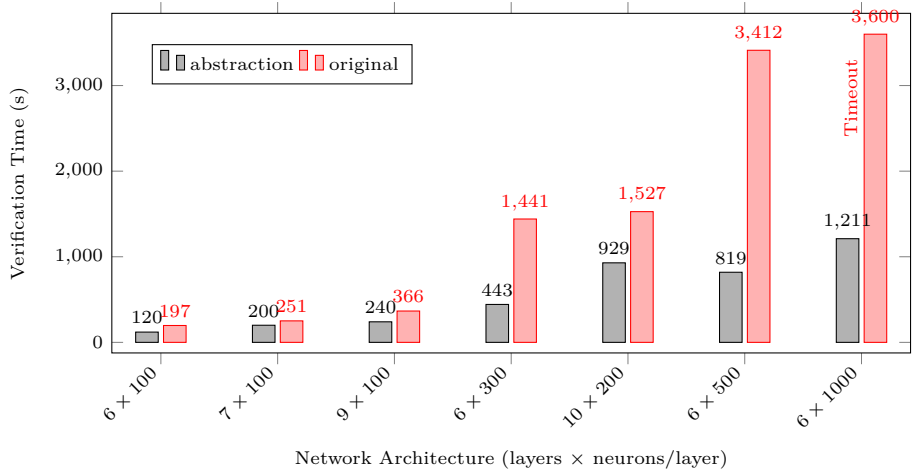
\begin{figure}
	\centering
	\begin{tikzpicture}
	\begin{axis}[
	font=\scriptsize,
	ybar,
	x label style={at={(axis description cs:0.5,-0.25)},anchor=south},
	xlabel={Network Architecture (layers $\times$ neurons/layer)},
	width=\textwidth,
	height=0.5\textwidth,
	enlargelimits=0.07,
	legend style={at={(0.05,0.85)}, anchor=west, legend columns=-1},
	ylabel={Verification Time (s)},
	symbolic x coords={$6\times100$,$7\times100$,$9\times100$,$6\times300$,$10\times200$,$6\times500$,$6\times1000$},
	xtick=data,
	nodes near coords,
	nodes near coords align={vertical},
	x tick label style={rotate=45,anchor=east},
	bar width=0.3cm,
	]
	\addplot[color=black,fill=black!30!white]
	coordinates{($6\times100$,120)($7\times100$,200)($9\times100$,240)($6\times300$,443)($10\times200$,929)($6\times500$,819)($6\times1000$,1211)};
	\addplot[color=red,fill=red!30!white]
	coordinates %
	{($6\times100$,197)($7\times100$,251) ($9\times100$,366) ($6\times300$,1441) ($10\times200$,1527) ($6\times500$,3412)($6\times1000$,3600)};
	\node [rotate=90, anchor=center, color=red] at (axis description cs: 0.925,0.75) {Timeout};
	\legend{abstraction,original}
	\end{axis}
	\end{tikzpicture}
	\caption{Accelerated verification after abstracting compared to verification of the original. The abstracted NN are verified directly without performing proof lifting, as if they were to replace the original one in their application. The time taken for abstracting (not included in the verification time) is 14, 14, 20, 32, 37, 53, and 214s respectively.}    \label{fig:verification-improvement}
\end{figure}

As mentioned in the Section \ref{sec:intro}, we found that the abstraction, considered as a standalone network, is faster to verify than the original network.
This opens up the possibility of verifying the abstraction and replacing the original network with it, in real-use scenarios.
In Figure \ref{fig:verification-improvement}, we show the time it takes to verify the abstract network using DeepPoly against the time taken to verify the respective original network.
Note that the reduction rate and accuracy drop of the corresponding networks can be found in Table \ref{tab:inc_number} above.
Clearly, there is a significant improvement in the run time of the verification algorithm; for the $6 \times 1000$ case, the verification algorithm timed out after 1 hour on the original network while it finished in less than 21 minutes on the abstract network.

\subsection{Results on lifting verification proof} \label{sec:verification-results}

\begin{table}[t]
	\centering
	\hfill
	\setlength{\tabcolsep}{12pt}
	\caption{Results of abstraction, verification and proof lifting of a $6\times300$ NN on 200 images to verify. The first column gives the number of neurons removed in layers 3, 4, 5 and 6 respectively. The second column shows the reduction in the size of the abstracted network compared to the original. We also report the number of images for which the original network could be proved robust by lifting the verification proof.}
	\label{tab:prooflifting}
	\begin{tabular}{lrrr}
		\toprule
		\textbf{Removed} & \textbf{Reduction} & \textbf{Images}   & \textbf{Verification}   \\
		  \textbf{Neurons} & \textbf{Rate (\%)} & \textbf{Verified} &   \textbf{Time ($min$)}   \\ \midrule
		   15, 25, 100, 100  &              13.33 &               195 &                    36 \\
		   15, 50, 100, 100  &              14.72 &               195 &                    36 \\
		   25, 25, 100, 100  &              13.89 &               190 &                    36 \\
		   25, 50, 100, 100  &              15.28 &               190 &                    36 \\
		   25, 100, 100, 100 &              18.06 &                63 &                    35 \\
		   50, 100, 100, 100 &              19.44 &                 0 &                    34\\\bottomrule
	\end{tabular}
\end{table}

Finally, we ran experiments to demonstrate the working of the full verification pipeline --- involving clustering to identify the neurons that can be merged, performing the abstraction (Section \ref{sec:clustering-abstraction}), running DeepPoly on the abstraction and finally lifting the verification proof to answer the verification query on original network (Section \ref{sec:verification}).

We were interested in two parameters: (i) the time taken to run the full pipeline; and (ii) the number of verification queries that could be satisfied (out of 200).
We ran experiments on a $6\times300$ network that could be verified to be locally robust for 197/200 images in 48 minutes by DeepPoly. The results are shown in Table \ref{tab:prooflifting}.
In the best case, our preliminary implementation of the full pipeline was able to verify robustness for 195 images in 36 minutes --- 13s for clustering and abstracting, 35 min for verification, and 5s for proof lifting.
In other words, a 14.7\% reduction in network size produced a 25\% reduction in verification time.
When we pushed the abstraction further, e.g. last row of Table \ref{tab:prooflifting}, to obtain a reduction of 19.4\% in the network size, DeepPoly could still verify robustness of the abstracted network for 196 images in just 34 minutes (29\% reduction). 
However, in this case, the proof could not be lifted to the original network as the over-approximations we obtained were too coarse.

This points to the interesting fact that the time taken in clustering and proof lifting are indeed not the bottlenecks in the pipeline. Moreover, a decrease in the width of the network indeed tends to reduce the verification time.
This opens the possibility of spending additional computational time exploring more powerful heuristics (e.g. principal component analysis) in place of the na\"ive $k$-means clustering in order to find smaller abstractions. Moreover, a counterexample-guided abstraction refinement (CEGAR) approach can be employed to improve the proof lifting by tuning the abstraction where necessary.

%% file: appendix.tex
\appendix

\section{Technical Details and Proofs}

\subsection{Correctness of merging} \label{app:hard-merge-correctness}

Suppose for every input, the activation values of two neurons $p$ and $q$ of layer $\ell$ are equal, i.e. $\z^{(\ell)}_p = \z^{(\ell)}_q$ 
(note that for the sake of readability, we omit the input from our equations and write $\z^{(\ell)}_i$ instead of $\z^{(\ell)}_i(x)$), 
then we argue that the neurons could be merged by keeping, without loss of generality, only neuron $p$ and setting the outgoing weights of $p$ to the sum of outgoing weights of $p$ and $q$. 
More formally, suppose the activation value of neuron $p$, $\z^{(\ell)}_p = \phi(\matr{W}{\ell-1}_{p,*} \vecz^{(\ell-1)} + \vecb^{\ell}_p)$, 
and that of neuron $q$, $\z^{(\ell)}_q = \phi(\matr{W}{\ell-1}_{q,*} \vecz^{(\ell-1)} + \vecb^{\ell}_q)$ are equal for every input to the network.
Let $\tilde\NN$ be the neural network obtained after merging neurons $p$ and $q$ in layer $\ell$. 
Note that $\NN$ and $\tilde\NN$ are identical in all layers which follow layer $\ell$. 
Due to the feedforward nature of the networks, it is easy to see that if for each input the vector of pre-activations of layer $\ell+1$ in $\NN$ and $\tilde\NN$ are same, 
i.e. $\tilde\vech^{(\ell+1)} = \vech^{(\ell+1)}$, then the outputs of $\NN$ and $\tilde\NN$ will also be the same.

The weight matrices of $\NN$ are copied to $\tilde\NN$. $\matr{\tilde W}{\ell-1}$ is set to $\matr{W}{\ell-1}$ with the $q^{th}$ row deleted. 
Further, we set $\matr{\tilde W}{\ell}_{*,p} = \matr{W}{\ell}_{*,p} + \matr{W}{\ell}_{*,q}$. 
Intuitively, this is same as deleting neuron $q$ and moving all its outgoing edges to neuron $p$. 
Suppose the pre-activation value of neuron $i$ of layer $\ell + 1$ of $\NN$ was given by

\begin{align*}
\h^{(\ell+1)}_i &= \vecb^{(\ell+1)}_i + \weight{\ell}{i,p}\z^{(\ell)}_p + \weight{\ell}{i,q}\z^{(\ell)}_q + \sum_{k \in \{1,\dots,n_\ell\}\setminus\{p, q\}}\weight{\ell}{i,k}\z^{(\ell)}_k\\
\noalign{\quad\ \ Since we assume that $\z^{(l)}_p = \z^{(l)}_q$, we can rewrite the RHS of the above equation as}
\h^{(\ell+1)}_i &= \vecb^{(\ell+1)}_i + (\weight{\ell}{i,p}+\weight{\ell}{i,q})\z^{(\ell)}_p + \sum_{k \in \{1,\dots,n_\ell\}\setminus\{p, q\}}\weight{\ell}{i,k}\z^{(\ell)}_k
\end{align*}

In the transformed NN $\tilde\NN$, since we have set $\matr{\tilde W}{\ell}_{*,p} = \matr{W}{\ell}_{*,p} + \matr{W}{\ell}_{*,q}$, we obtain
\begin{align*}
\tilde\h^{(\ell+1)}_i &= \vecb^{(\ell+1)}_i + (\weight{\ell}{i,p}+\weight{\ell}{i,q})\z^{(\ell)}_p + \sum_{k \in \{1,\dots,n_\ell\}\setminus\{p, q\}}\weight{\ell}{i,k}\z^{(\ell)}_k
\end{align*}

\subsection{Error bounds} \label{app:error-bounds}

Let $n_\ell$ denote the number of neurons in layer $\ell$. We use the symbol $\vecz^{(\ell)} = [\z^{(\ell)}_1, \dots, \z^{(\ell)}_{n_\ell}]^\intercal$ to denote the column vector of activations of layer $\ell$, 
$\matr{W}{\ell} = (\weight{\ell}{ji})$ to denote the $n_{\ell+1} \times n_{\ell}$ matrix of weights $\weight{\ell}{ji}$ of the edge from node $i$ in layer $\ell$ to node $j$ in layer $\ell+1$, 
$\vech^{(\ell)} = [\h^{(\ell)}_1, \dots, \h^{(\ell)}_{n_\ell}]^\intercal$ denotes the column vector of pre-activations of layer $\ell$,
and $\vecb^{(\ell)} = [\b^{(\ell)}_1, \dots, \b^{(\ell)}_{n_\ell}]^\intercal$ to denote the column vector of biases of layer $\ell$.

In the rest of the discussion, we omit the parameter $x$ and write $\vecz^{(\ell)}$ or $\vech^{(\ell)}$ instead of $\vecz^{(\ell)}(x)$ or $\vech^{(\ell)}(x)$ respectively for the sake of readability.

\begin{lemma}[Single-step error] \label{lem:del-lplus1}
	If the activations $\vecz^{(\ell)}$ of a single layer $\ell$ are perturbed by $\Delta \vecz^{(\ell)}$, then the perturbation of the activations of layer $\ell+1$ is bounded according to 
	$$\Delta \vecz^{(\ell+1)} \leq \phi(\matr{W}{\ell}\Delta \vecz^{(\ell)})$$
	if the activation function $\phi$ is sub-additive.
\end{lemma}
\begin{proof}[of Lemma \ref{lem:del-lplus1}]
	Suppose that $\vecz^{(\ell)}$ was perturbed by some $\Delta \vecz^{(\ell)}$ to obtain the new activation $\tilde\vecz^{(\ell)} = \vecz^{(\ell)} + \Delta \vecz^{(\ell)}$, then we would define
	\begin{align}
	\tilde\vech^{(\ell+1)} &= \matr{W}{\ell} \tilde\vecz^{(\ell)}  + \vecb^{(\ell+1)}\notag
	\end{align}
	
	following which, we can bound the difference between the original $\vech^{(\ell+1)}$ and the perturbed $\tilde\vech^{(\ell+1)}$:
	\begin{align}
	\Delta \vech^{(\ell+1)} = \tilde\vech^{(\ell+1)} - \vech^{(\ell+1)} &= \matr{W}{\ell} ( \tilde\vecz^{(\ell)} - \vecz^{(\ell)} )\notag\\
	&= \matr{W}{\ell} \Delta \vecz^{(\ell)}\notag
	\end{align}
	
	When this error is propagated across the neurons of the $(\ell+1)^{th}$ layer, we have
	\begin{align}
	\tilde\vecz^{(\ell+1)} = \phi(\tilde\vech^{(\ell+1)}) &= \phi(\vech^{(\ell+1)} + \Delta \vech^{(\ell+1)}) \notag\\
	&= \phi(\vech^{(\ell+1)} + \matr{W}{\ell} \Delta \vecz^{(\ell)}) \label{eq:tilde-z-l-plus-1}
	\end{align}
	If $\phi$ is sub-additive, we have
	\begin{align}
	\tilde\vecz^{(\ell+1)} &\leq \phi(\vech^{(\ell+1)}) + \phi(\matr{W}{\ell} \Delta \vecz^{(\ell)})\notag\\
	\tilde\vecz^{(\ell+1)} &\leq \vecz^{(\ell+1)} + \phi(\matr{W}{\ell} \Delta \vecz^{(\ell)})\notag\\
	\Delta \vecz^{(\ell+1)} &\leq \phi(\matr{W}{\ell} \Delta \vecz^{(\ell)}) \notag
	\end{align}
	\qed
\end{proof}

\begin{proof}[of Theorem \ref{thm:accum-abs-error}]
	Assume we already have clustered all layers up to layer $\ell+1$ and we know the accumulated error for layer $\ell$, namely $\accumerror^{(\ell)}$. 
	The error in layer $\ell+1$ is defined as $\accumerror^{(\ell+1)}=\lvert\zpost^{(\ell+1)}-\vecz^{(\ell+1)}\rvert$, where $\zpost$ denotes the activation values of layer $\ell+1$ after clustering it. Let $\zpre^{(\ell+1)}$ denote the activation values of layer $\ell+1$ when all layers before are clustered but not the layer itself, and $\vecz^{(\ell+1)}$ shall be the original activation values. We have
	\begin{align}\label{eq:acc-error}
	\lvert\accumerror^{(\ell+1)}\rvert&=\lvert\zpost^{(\ell+1)}-\vecz^{(\ell+1)}\rvert\\
	&=\lvert\zpost^{(\ell+1)}-\zpre^{(\ell+1)}+\zpre^{(\ell+1)}-\vecz^{(\ell+1)}\rvert\\
	&\leq \lvert\zpost^{(\ell+1)}-\zpre^{(\ell+1)}\rvert+\lvert \zpre^{(\ell+1)}-\vecz^{(\ell+1)}\rvert
	\end{align}

	We know from Lemma \ref{lem:del-lplus1} how the error is propagated to the next layer. So, we know
	\begin{equation}
		\lvert \zpre^{(\ell+1)}-\vecz^{(\ell+1)}\rvert\leq\lvert\phi(\matr{W}{(\ell)}\accumerror^{(\ell)})\rvert
	\end{equation}
	We now have to consider the error introduced in layer $\ell+1$ by the clustering.
	From definition, it is $|\vecz_{r_i}-\vecz_i|\leq\vec\epsilon_{r_i}$ for any node $i$ and its cluster representative $r_i$. Note that any node is contained in a cluster but that most of the clusters have size 1. For most nodes, we would then have $i=r_i$. However, in the general case we get
	\begin{equation}
	\lvert\zpost^{(\ell+1)}-\zpre^{(\ell+1)}\rvert\leq\vec\epsilon^{(\ell+1)}
	\end{equation} 
	Thus, equation \ref{eq:acc-error} becomes
	\begin{equation}
		\lvert\accumerror^{(\ell+1)}\rvert\leq\lvert\phi(\matr{W}{(\ell)}\accumerror^{(\ell)})
		\rvert+\vec\epsilon^{(\ell+1)}
	\end{equation}
	This can be made simpler for the ReLU-, parametric or leaky ReLU and the tanh-activation function. For all of them, it holds $|\phi(x)|\leq|x|$. Thus
	\begin{equation}
	\lvert\accumerror^{(\ell+1)}\rvert\leq\lvert\matr{W}{(\ell)}\accumerror^{(\ell)}
	\rvert+\vec\epsilon^{(\ell+1)}
	\end{equation}
	which is what we wanted to show.
	\qed

\end{proof}

\begin{proof}[of Theorem \ref{thm:robustness}]
	We are interested in computing $ \lvert\accumerror_{total}\rvert=\vert\tilde\NN(\tilde x) - \NN(x)\rvert$ which can be rewritten as $\lvert\tilde\NN(\tilde x) - \tilde\NN(x)) + (\tilde\NN(x) - \NN(x)\rvert$.\\
	$\lvert\tilde\NN(\tilde x) - \tilde\NN(x)\rvert \leq \lvert \matr{\tilde{W}}{L} \dots \matr{\tilde{W}}{1} \perturb  \rvert$ is a consequence of Lemma \ref{lem:del-lplus1} and under the assumption that the activation function $\phi$ fulfills $\phi(x)\leq x$, which is true for ReLU and tanh.\\
	$\lvert\tilde\NN(x) - \NN(x)\rvert = \lvert\accumerror^{(L)}\rvert$ which is a direct consequence of Theorem \ref{thm:accum-abs-error}.
	\qed
\end{proof}

\subsection{Lifting guarantees}\label{app:lifting-proof}
\begin{proof}[of Theorem \ref{thm:monster-thm}]
	As the verification of a specific property only considers the upper and lower bound of the output layer $L$, it is sufficient to show that  $\uover(L)\geq\uorig(L)$ and $\lover(L)\leq\lorig(L)$, where $\uorig$ and $\lorig$ correspond to the upper- and lower-bound given by DeepPoly on the original network, and $\uover(L)$ and $\lover(L)$ denote the over-approximations.\\
	We can show this inductively, where the base case is obvious. For the first layer, $\uover^{(1)} = \upre^{(1)} + \uaccdelta(1) = \uorig^{(1)} + 0$ and $\lover^{(1)} = \lpre^{(1)} - \laccdelta(1) = \lorig^{(0)} - 0$.\\
	Let's consider now some layer $\ell$ and start with the upper bound. We have
	\begin{align}
		\uorig^{(\ell)}=\max\Bigg(0,\matr{W}{\ell-1}_+u^{(\ell-1)}+\matr{W}{\ell-1}_-l^{(\ell-1)}+b^\ell\Bigg)
	\end{align}
	from the calculation of \cite[Section 4.4]{deeppoly} and 
	\begin{align}
		\uover^{(\ell)}=\max\Bigg(0, \matr{\tilde{W}}{\ell-1}_+(\uover^{(\ell-1)}+\vec\epsilon^{(\ell-1)})+\matr{\tilde{W}}{\ell-1}_-(\lover^{(\ell-1)}-\vec\epsilon^{(\ell-1)})+\tilde{b}^\ell\Bigg)
	\end{align}
	by our definition.\\
	We want to show that $\uover^{(\ell)}-\uorig^{(\ell)}\geq 0$. We can leave out the $\max$ operation, because it is clear that $(a-b\geq0)\Rightarrow(\max(0,a)-\max(0,b)\geq0)$.\\
	Let's consider only one node in layer $\ell$, say node $n$, and omit the $\max$-operation. For simplicity, we also omit the superscript $\ell-1$ in the following calculation. Let $I$ denote all nodes from layer $\ell-1$ in the original network and $\tilde{I}$ in the abstracted one. We get
	\begin{align*}
		\uover_n^{(\ell)}-\uorig_n^{(\ell)} &= \sum_{i\in \tilde{I}} \tilde{w}^+_{i,n}(\hat{u}_i+\epsilon_i)\\
		&+ \sum_{i\in \tilde{I}} \tilde{w}^-_{i,n}(\hat{l}_i-\epsilon_i)\\
		&- \Bigg(\sum_{i\in I} w^+_{i,n}u_i
		+ \sum_{i\in I}w^-_{i,n}l_i\Bigg)\\
	\end{align*}
	It is $\tilde{I}\subset I$ and we can map all nodes in $I$ to their corresponding cluster $c$ with its cluster-representative $r\in\tilde{I}$. Thus we get
	\begin{align*}
		\uover_n^{(\ell)}-\uorig_n^{(\ell)}=&
		\sum_{c\,\text{cluster}}\Bigg(\tilde{w}^+_{r,n}(\hat{u}_r+\epsilon_r)-\sum_{m\in c}w^+_{m,n}u_m\Bigg)\\
		+&\sum_{c\,\text{cluster}}\Bigg(\tilde{w}^-_{r,n}(\hat{l}_r-\epsilon_r)-\sum_{m\in c}w^-_{m,n}l_m\Bigg)\\
	\end{align*}
	For each cluster $c$, there are two cases: either it contains only one node or more than one node. In the case of one node per cluster, the abstracted network does not differ from the original one, so $\tilde{w}_r = w_r$. For such cluster $c$, we get
	\begin{align*}
		&w^+_{r,n}(\hat{u}_r+\epsilon_r)-w^+_{r,n}u_r\\
		= &\, w^+_{r,n}(\hat{u}_r+\epsilon_r-u_r)\\
		= &\, w^+_{r,n}(\hat{u}_r-u_r)\\
		\geq & \,0
	\end{align*}
	and
	\begin{align*}
		&w^-_{r,n}(\hat{l}_r-\epsilon_r)-w^-_{r,n}l_r\\
		= &\, w^-_{r,n}(\hat{l}_r-\epsilon_r-l_r)\\
		= &\, w^-_{r,n}(\hat{l}_r-l_r)\\
		\geq & \,0
	\end{align*}
	because for any un-clustered node $\epsilon_r=0$, and by induction hypothesis $\hat{u}_r\geq u_r$ and $\hat{l}_r\leq l_r$.\\
	
	Let's now consider the second case, where one cluster contains more than one node. In such cluster $c$, we have $\tilde{w}_r=\sum_{m\in c}w_m$, so
	\begin{align*}
		&\tilde{w}^+_{r,n}(\hat{u}_r+\epsilon_r)-\sum_{m\in c}w^+_{m,n}u_m\\
		=&\,\sum_{m\in c}w^+_{m,n}(\hat{u}_r+\epsilon_r)-\sum_{m\in c}w^+_{m,n}u_m\\
		=&\,\sum_{m\in c}w^+_{m,n}(\hat{u}_r+\epsilon_r-u_m)\\
		\geq & \,0
	\end{align*}
	and
	\begin{align*}  
		&\tilde{w}^-_{r,n}(\hat{l}_r-\epsilon_r)-\sum_{m\in c}w^-_{m,n}l_m)\\
		=&\,\sum_{m\in c}w^-_{m,n}(\hat{l}_r-\epsilon_r-l_m)\\
		\geq & \,0
	\end{align*}
	because by definition of $\epsilon_r$, we have for all $m\in c$: $\hat{u}_r+\epsilon_r\geq u_m$ and similarly, $\hat{l}_r-\epsilon_r\leq l_m$.
	We get that $\uover_n^{(\ell)}-\uorig_n^{(\ell)}\geq 0$. The calculation for the lower bounds follows the same principle just with exchanged signs and is thus not presented here.
	\qed
\end{proof}

\subsection{Details on training process}\label{app:training-details}

We generated various NN architectures by scaling up the number of neurons per layer as well as the number of layers themselves and trained them on MNIST. For doing so, we split the dataset into three parts: one for the training, one for validation and one for testing. The training is then performed on the training dataset by using common optimizers and loss functions. The training was stopped when the accuracy on the validation set did not increase anymore.

The NN on MNIST were trained on 60000 samples from the whole dataset. Of these, 10\% are split for validation, thus there are 54000 images for the training itself and 6000 images for validation.
The optimizer used for the training process is ADAM, which is an extension to the stochastic gradient descent. To prevent getting stuck in local minima, it includes the first and second moments of the gradient. It is a common choice for the training of NN and performed reasonably well in this application. Its parameter are set to the default from TensorFlow, namely a learning rate of $0.001$, $\beta_1=0.9$, $\beta_2=0.999$ and $\epsilon=1e-07$.\\
For MNIST, the most reasonable loss function is the sparse categorical crossentropy.
The training process was stopped when the loss function on the validation data did not decrease anymore. Usually, the process would stop after at most 10 epochs.

\subsection{Proof lifting example}\label{app:example-lifting}

\begin{figure}
	\centering
	\begin{tikzpicture}[node distance=2.1cm, every node/.style={circle}, removed/.style={gray}, font=\scriptsize, execute at begin node=\setlength{\baselineskip}{12pt}]
	\node[draw, label={[label distance=1pt, align=left]$\lpre_1 = -1$\\$\upre_1 = 1$}] (1) {$\xpre_1$};
	\node[draw, label={[label distance=-16pt, below,align=left]$\lpre_2 = -1$\\$\upre_2 = 1$}, below of =1] (2) {$\xpre_2$};
	\node[draw,  label={0},label={[label distance=1pt, align=left]$\lpre_3 = -2$\\$\upre_3 = 2$}, right of =1] (3) {$\xpre_3$};
	\node[draw,  label={0},label={[label distance=-16pt, below,align=left]$\lpre_4 = -2$\\$\upre_4 = 2$}, below of =3] (4) {$\xpre_4$};
	\node[draw, label={[label distance=3pt, align=left]$\lpre_5 = 0$\\$\upre_5 = 2$}, right of =3] (5) {$\xpre_5$};
	\node[draw, label={[label distance=-18pt, below,align=left]$\lpre_6 = 0$\\$\upre_6 = 2$}, below of =5] (6) {$\xpre_6$};
	\node[draw, label={0}, label={[label distance=4pt, align=left]$\lpre_7 = 0$\\$\upre_7 = 4$}, right of =5] (7) {$\xpre_7$};
	\node[removed, draw, label={0}, below of =7] (8) {$\xpre_8$};
	\node[draw, label={[label distance=4pt, align=left]$\lpre_9 = 0$\\$\upre_9 = 4$}, right of =7] (9) {$\xpre_9$};
	\node[removed, draw, below of =9] (10) {$\xpre_{10}$};
	\node[draw, label={5}, label={[label distance=0pt, align=left]$\lpre_{11} = 5$\\$\upre_{11} = 13$}, right of =9] (11) {$\xpre_{11}$};
	\node[draw, label={0}, label={[label distance=-18pt, below,align=left]$\lpre_{12} = 0$\\$\upre_{12} = 4$}, below of =11] (12) {$\xpre_{12}$};
	
	\draw[->] (1) -- (3) node[midway,above] {1};
	\draw[->] (1) -- (4) node[near end,above] {1};
	\draw[->] (2) -- (3) node[near end,above] {1};
	\draw[->] (2) -- (4) node[midway,above] {-1};
	
	\draw[->] (3) -- (5) node[midway,above=-14pt] {$\max(0, \xpre_3)$};
	\draw[->] (4) -- (6) node[midway,above=-14pt] {$\max(0, \xpre_4)$};
	
	\draw[->] (5) -- (7) node[midway,above] {1};;
	\draw[removed,->] (5) -- (8) node[near end,above] {\cancel{1}};
	\draw[->] (6) -- (7) node[near end,above] {1};
	\draw[removed,->] (6) -- (8) node[midway,above] {\cancel{1}};
	
	\draw[->] (7) -- (9) node[midway,above=-14pt] {$\max(0, \xpre_7)$};
	\draw[removed,->] (8) -- (10) node[midway,above=-14pt] {$\max(0, \xpre_8)$};
	
	\draw[->] (9) -- (11) node[midway,above] {\textcolor{gray}{\cancel{1}} \textcolor{purple}{2}};
	\draw[->] (9) -- (12) node[near end,above] {\textcolor{gray}{\cancel{0}} \textcolor{purple}{1}};
	\draw[removed,->] (10) -- (11) node[near end,above] {\cancel{1}};
	\draw[removed,->] (10) -- (12) node[midway,above] {\cancel{1}};
	\end{tikzpicture}
	\caption{Abstracted network showing the constrains returned by DeepPoly. Note that this is equivalent to having the ReLU unit identified by neurons 8 and 10 merged into the ReLU unit identified by neurons 7 and 9.}
	\label{fig:deeppoly-2}
\end{figure}
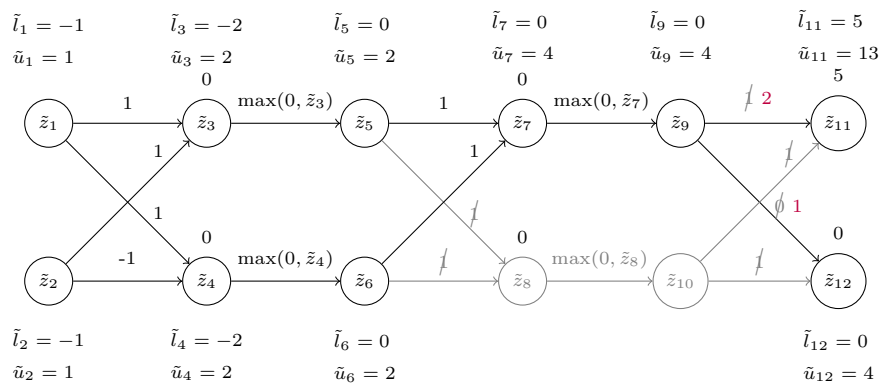

\begin{example}
	Consider the network shown in Figure \ref{fig:deeppoly-2}. The ReLU layers have already been split into two: those (i) computing the affine sum and (ii) computing $\max(0, \cdot)$. The greyed/striked out weights belong to the original network but are not present in the abstract network, in which the ReLU unit identified by neurons 8 and 10 have been merged into the ReLU unit identified by neurons 7 and 9. The two weights coloured purple (between neuron 9 and 11 as well as between 9 and 12) are a result of abstraction, as described in Section \ref{sec:merge}.
	
	We apply the DeepPoly algorithm as discussed in \cite[Section 2]{deeppoly} on the abstracted network to obtain the bounds shown in the figure.
	Neurons 1-7 and 9 are unaffected by the merging procedure.
	For neuron 11, we have $\lpre_{11} = 5$ and $\upre_{11} = 13$, and for neuron 12, $\lpre_{12} = 0$ and $\upre_{12} = 4$.
	Since $\lpre_{11} > \upre_{12}$, we can conclude that the abstracted network is robust, however, we do not know if the lower bound $\lorig_{11}$ from the original network is indeed greater than the upper bound $\uorig_{12}$.
	
	Using Theorem \ref{thm:monster-thm} and the result of DeepPoly on the abstraction, we can compute the bounds $[\lover, \uover]$ such that it contains $[\lorig, \uorig]$, the bounds that would have been computed by DeepPoly for the original network.

	\begin{align*}
	\lover^{(6)} &\leq \lorig^{(6)} \\
	\text{or} \quad\quad \begin{bmatrix}
	\lover_{11} \\
	\lover_{12}
	\end{bmatrix} &\leq \begin{bmatrix}
	\lorig_{11}\\
	\lorig_{12}
	\end{bmatrix}
	\end{align*}
	and
	\begin{align*}
	\uover^{(6)} &\geq \uorig^{(6)} \\
	\text{or} \quad\quad \begin{bmatrix}
	\uover_{11} \\
	\uover_{12}
	\end{bmatrix} &\geq \begin{bmatrix}
	\uorig_{11}\\
	\uorig_{12}
	\end{bmatrix}
	\end{align*}
	
	Assuming the cluster diameter to be $\epsilon$ as defined in Equation \ref{eq:max-cluster-width}, we get $$\uover(6) = \begin{bmatrix}
	13+2\epsilon^{(4)}\\
	4+\epsilon^{(4)}
	\end{bmatrix}$$ and 
	$$\lover(6) = \begin{bmatrix}
	5-2\epsilon^{(4)}\\
	-\epsilon^{(4)}
	\end{bmatrix}$$
	 and therefore,
	\begin{align*}
	\begin{bmatrix}
	5-2\epsilon^{(4)}\\
	-\epsilon^{(4)}
	\end{bmatrix}
	\leq \begin{bmatrix}
	\lorig_{11}\\
	\lorig_{12}
	\end{bmatrix}
	\text{ and }
	\begin{bmatrix}
	\uorig_{11}\\
	\uorig_{12}
	\end{bmatrix}\leq
	\begin{bmatrix}
	13+2\epsilon^{(4)}\\
	4+\epsilon^{(4)}
	\end{bmatrix}
	\end{align*}
	To determine if t $5-2\epsilon^{(4)}=\lover_{11} > \uover_{12}=4+\epsilon^{(4)}$ holds, we need to have a value for $\epsilon^{(4)}$. As this is only a toy example and the neurons were chosen manually and not by clustering, we do not have a value for it here. 
	However, one can see that the proof lifting heavily depends on this value. If it was $\epsilon^{(4)}\geq\frac{1}{3}$, the property could not be lifted, even it would theoretically hold on the original network.
\end{example}

\subsection{Implementation details}\label{app:implementation-details}
We implemented the abstraction technique described in Section \ref{sec:clustering-abstraction} using the popular deep learning library TensorFlow \cite{tensorflow2015-whitepaper} and the machine learning library Scikit-learn \cite{scikit-learn}. For the verification, we used the DeepPoly implementation available in the ERAN toolbox\footnote{Available at \href{https://github.com/eth-sri/ERAN}{github.com/eth-sri/ERAN}}